\documentclass[11pt]{article}

\usepackage{array}
\usepackage{amsfonts,amsmath,amssymb,amsthm,mathrsfs}
\usepackage{adjustbox}
\usepackage{tikz}
\usepackage{enumitem}
\usepackage{ifthen}
\usepackage{fullpage}

\usepackage[ruled]{algorithm2e} 

\SetAlFnt{\small}
\SetAlCapFnt{\small}
\SetAlCapNameFnt{\small}
\SetAlCapHSkip{0pt}
\IncMargin{-\parindent}

\usepackage{booktabs} 
\usepackage{xspace}
\usepackage[pagebackref=true,hidelinks]{hyperref}
\usepackage{bbm}
\usepackage{enumitem}
\usepackage[numbers]{natbib}
\usepackage{multirow}
\usepackage{graphicx}
\usepackage{subfigure}
\usepackage{wrapfig}
\usepackage{cleveref}

\newtheorem{theorem}{Theorem}[section]
\newtheorem{proposition}{Proposition}[section]
\newtheorem{lemma}[theorem]{Lemma}
\newtheorem{corollary}[theorem]{Corollary}

\usepackage[normalem]{ulem}
\usepackage{xcolor}
\usepackage{bm}

\newcommand\ytnote[1]{}
\newcommand\rpnote[1]{}
\newcommand\bsnote[1]{}

\newcommand{\mech}{\mathcal{M}}
\newcommand{\flow}{\lambda}
\newcommand{\R}{\mathbb{R}}
\newcommand{\E}{\mathbb{E}}
\newcommand{\bids}{\textbf{b}}
\newcommand{\alloc}{\textbf{x}}
\newcommand{\pay}{\textbf{p}}
\newcommand{\bidf}{\beta}
\newcommand{\lbidf}{\beta^-}
\newcommand{\rbidf}{\beta^+}
\newcommand{\bidfn}{\beta^*}
\newcommand{\values}{\textbf{v}}
\newcommand{\bidfs}{\bm{\beta}}
\newcommand{\rev}{\normalfont\textsc{Rev}}
\newcommand{\fp}{\normalfont\textrm{FP}}
\newcommand{\secp}{\normalfont\textrm{SP}}
\newcommand{\avgp}{\bar{p}}

\newcommand{\C}{\mathcal{C}}
\newlength{\doubletextwidth}
\begin{document}
\title{Why Do Competitive Markets Converge to First-Price Auctions?}
\author{
Renato Paes Leme \\ Google Research \\ {\tt renatoppl@google.com} \and 
Balasubramanian Sivan \\ Google Research \\ {\tt balusivan@google.com} \and 
Yifeng Teng \\ UW-Madison \\ {\tt yifengt@cs.wisc.edu}
}
\date{}

\maketitle
\thispagestyle{empty}

\begin{abstract}
We consider a setting in which bidders participate in multiple auctions run by different sellers, and optimize their bids for the \emph{aggregate} auction. We analyze this setting by formulating a game between sellers, where a seller's strategy is to pick an auction to run. Our analysis aims to shed light on the recent change in the Display Ads market landscape: here, ad exchanges (sellers) were mostly running second-price auctions earlier and over time they switched to variants of the first-price auction, culminating in Google's Ad Exchange moving to a first-price auction in 2019. Our model and results offer an explanation for why the first-price auction occurs as a natural equilibrium in such competitive markets.

\end{abstract}

\setcounter{page}{1}

\section{Introduction}

The research questions investigated in this paper arise in the backdrop of numerous ad exchanges in the display ads market having switched to first-price auctions in the recent few years, culminating in Google Ad Exchange's move to first-price auction in September 2019~\cite{googFPA}.
This is well summarized in this quote from Scott Mulqueen in~\cite{blogFPA}: ``\emph{Moving to a first-price auction puts Google at parity with other exchanges and SSPs in the market, and will contribute to a much fairer transactional process across demand sources.}''

In this paper we study the display ads market as a game between ad exchanges where the strategy of each exchange is the choice of auction mechanism it picks. Our goal in analyzing this game is to shed light on the incentives for an exchange to prefer one auction over another, and in particular, study if there is a strong game-theoretic justification for all exchanges converging to a first-price auction. Why is the first-price auction (FPA) special?

In our model we consider a display ads market with $n$ sellers $N = \{1,2,\dots,n\}$ representing the ad exchanges, and $m$ buyers representing the bidding networks. Every day each of these exchanges runs billions of auctions, with buyer values drawn independently from a distribution $F$ in each auction. Exchange $j$ controls $\lambda_j$ fraction of the auctions, i.e., any single query could originate from exchange $j$ with probability $\lambda_j$. Every exchange $j$ chooses the mechanism $\mech_j$ it runs. The buyers then decide on a bidding strategy mapping values to bids and use the same strategy \emph{across all exchanges}, i.e, they bid in equilibrium for the average auction $\sum_j \lambda_j \mech_j$.

In this competition between exchanges, what will be the outcome? Is there an equilibrium in this game between exchanges? If so, over what class of allowed auctions? Is the equilibrium symmetric? Is it unique? These are the central questions of the study in this paper. The answers to these questions are nuanced. The novelty in this competition is that when an exchange updates its auction, the change in its buyer response crucially depends on what the other exchanges were running and what their market powers were.

We begin by discussing a key modeling assumption we make before discussing results and interpretations.

\paragraph{Key assumption: Bidders respond to average auction:} Why do we assume that a bidder responds to the average auction rather than choosing a different bidding strategy for each exchange? To set the context: buyers derive value by maximizing their reach, and therefore usually buy inventory from numerous ad exchanges. This entails buyers bidding in numerous auctions simultaneously. While large sophisticated buyers devise exchange-specific strategies, smaller buyers often design a uniform strategy for the entire market. Practical reasons causing buyers not to have exchange specific strategies include: (a) bidders often track their Key Performance Indicators (KPIs) like impressions won, clicks, conversions, ad spend etc. across all exchanges, and don't maintain exchange-specific goals, leading to not having exchange specific strategies; (b) advertisers often participate in a preliminary auction within an ad network that represents the advertisers, and the ad network picks the winning advertiser's bid or clearing price in the preliminary auction and passes it transparently as the bid to the main auction run. This again tends to encourage uniform strategies across exchanges as the preliminary auction run by the ad network is not exchange specific, (c) often developing a bidding strategy for non-truthful auctions is costly and the gain in customizing it per exchange might not justify its cost, especially in cases where many exchanges use close enough auction formats. While we do not claim that every bidder is agnostic to the specific auction used in a query, the fact that a significant fraction of the advertisers are, justifies the model.

Another assumption we make is that the buyer values are drawn iid. Is this justified? Two relevant points here are (a) while iid assumption does not capture all scenarios (e.g. retargeting, where some advertiser has very high values for a returning user), it is not too removed from reality --- we note that the buyers are competing to impress an ad to capture the \emph{same user's attention} (b) even with the iid assumption, the main result is quite technically involved; when buyer values are non-iid, even a single exchange first-price auction's bidder equilibrium for general distributions is very hard to reason about and doesn't have a closed form; therefore gathering general insights about the equilibrium among exchanges would be rendered practically infeasible.

\paragraph{Auctions without reserves, revenue equivalence, and why is FPA sought after:} We start by analyzing a setting where each exchange always allocated to the highest bidder (i.e., no reserves), but allowed to charge the winner arbitrary payments smaller than bid (as long as losers pay $0$). This serves to illustrate both (a) why first-price auctions are universally sought after and (b) the central role of revenue equivalence theorem. When exchanges are allowed to run auctions within the class of auctions above, it turns out that revenue equivalence  simplifies the problem: the total of revenue of all exchanges, remains a constant \emph{regardless of what auction any exchange runs}. So the only question that remains is how the constant sized pie gets split between the exchanges. We argue that if not all exchanges are running first-price, at least one of those not running the first-price auction can capture a larger fraction of the pie by switching to a first-price auction. This fact drives the first-price auction being the \emph{unique Nash equilibrium} in this competition between exchanges.

\paragraph{Auctions with reserves, why revenue equivalence is useless, and why is FPA still sought after:} The simplicity afforded by the seemingly innocuous assumption of no reserves completely vanishes in the more realistic setting when one allows the exchanges to set reserve prices. Indeed, even reasoning about the setting where exchanges are only allowed to run first-price auctions, but are able to set arbitrary reserves is non-trivial (see Section~\ref{sec:simpleex} for an example). Several complications arise here, even with just two competing exchanges and two bidders. First, the two exchanges running first-price auctions with reserves $r_1$ and $r_2$ leads to a different outcome (different allocations, different payments, different sum of revenues of all exchanges) than two exchanges running second-price auctions with reserves $r_1$ and $r_2$ (compare this to the case where when exchanges are forced to use the same fixed reserve or have no reserve price at all, the auction choice does not matter). As a result, the total size of the revenue pie is no longer constant. Second, the equilibrium bidding function is not even continuous, with every unique reserve contributing an additional segment to the bidding function. Third, a symmetric pure strategy equilibrium among the exchanges ceases to exist in general. Fourth, the pure-strategy equilibrium is neither revenue optimal nor welfare optimal.

\paragraph{First-price auction's desirability is robust:} Notwithstanding all the complications listed earlier, we show that when each exchange is allowed to choose between a first-price auction with an arbitrary reserve of exchange's choice and a second-price auction with an arbitrary reserve of exchange's choice, every exchange will pick the first-price auction with reserve! Thus, the first-price auction's desirability is not limited to the case where we are able to apply revenue equivalence. The proof of this fact is quite involved and entails analyzing the different segments of the bidding functions to infer consequential properties.

While our analysis justifies a switch to first-price auctions from a game-theoretical angle, there are several other arguments supporting a first-price auction in general, e..g, the extra transparency / credibility offered by first-price~\cite{AL18}, the uniqueness of equilibrium~\cite{CH13} offered by first-price semantics compared to second-price semantics.

\paragraph{Related Work.} Auctions have been compared along various desiderata in the past. In their famous essay on ``The Lovely But Lonely Vickrey Auction''~\cite{AM06} Ausubel and Milgrom discuss how the Vickrey auction is rarely used in practice despite its academic importance, and how the very closely related ascending auction happens to be quite popular (Christie's, Sotheby's, eBay all use some variants of ascending auction). In a related paper~\cite{AM02}, Ausubel and Milgrom compare the VCG auction with the ascending package auctions in terms of the ability of bidders to understand, incentive properties, their equilibrium outcomes etc. However this is the first work comparing auctions from a revenue-performance-in-a-competitive-marketplace standpoint, and establish the superior performance of the first-price auction. 

Closely related to this paper is the stream of work on \emph{competing mechanism designers} by McAfee \cite{mcafee1993mechanism}, Peters and Serverinov \cite{peters1997competition}, Burguet and Sakovics \cite{burguet1999imperfect}, Pavan and Calzonari \cite{pavan2010truthful} and Pai \cite{Pai09}. We refer to the excellent survey by Mallesh Pai on the topic \cite{Pai10}. In McAfee's model \cite{mcafee1993mechanism} there is a (repeated) two-stage game where seller's first propose an auction mechanism and then unit-demand buyers decide in which mechanism they will participate in. The seller's choice of mechanism needs to balance between extracting revenue and attracting buyers to the auction. The crucial difference between this and our model is that buyers in our model want to buy as much inventory as possible and therefore bid on all mechanisms simultaneously. Instead of choosing one mechanism to participate, our buyers respond to the average mechanism.

\paragraph{Organization.} In Section~\ref{sec:prelim} we provide some preliminary details. In Section~\ref{sec:noreserve} we consider the case of auctions without reserves. In Section~\ref{sec:withreserve} we consider the case of auctions with reserves. In particular, in Section~\ref{sec:mainresults} we state our main results. In Section~\ref{sec:simpleex} we consider a simple example with two bidders and two exchanges to illustrate the complications that manifest themselves when we allow auctions with reserves. In Section~\ref{sec:thmproofs} we give the proofs of all our results. 
\section{Notations and Models}
\label{sec:prelim}
\subsection{Sealed-bid auctions and Bayes-Nash equilibrium}
In a single-item sealed-bid auction with $n$ bidders, we have each bidder's private value $v_i$ drawn from a publicly known distribution $F_i$. A mechanism $(\alloc,\pay)$ maps the bids $\bids = (b_1, \hdots, b_n)$ submitted by the bidders to allocation $x_i(\bids)$ which is the probability that bidder $i$ gets allocated the item, and payment $p_i(\bids)$ that bidder $i$ needs to pay.

Each bidder's strategy corresponds to a  bidding function $\bidf_i:\R\to\R$ to map their value to a bid. Denote by $\bidfs(\values)=(\bidf_1(v_1), \bidf_2(v_2), \cdots, \bidf_n(v_n))$ the vector of bids from all bidders with value profile $\values=(v_1,v_2,\cdots,v_n)$. The goal of each bidder $i$ is to maximize his (quasi-linear) utility $u_i(\bids)=x_i(\bids)v_i-p_i(\bids)$. We say that the bidders' bidding strategy profile $(\bidf_1,\cdots,\bidf_n)$ forms a Bayes-Nash equilibrium if under this profile no single bidder has an incentive to switch his bidding strategy, assuming they only know other bidders' value distributions and not true values. That is,
\begin{equation*}
    \E_{\values_{-i}\sim \bm{F}_{-i}}[u_i(\bidf_i(v_i),\bidfs_{-i}(\values_{-i}))]\geq \E_{\values_{-i}\sim \bm{F}_{-i}}[u_i(\gamma_i(v_i),\bidfs_{-i}(\values_{-i}))]
\end{equation*}
for any other bidding function $\gamma_i$.

\subsection{Competitive market as a two-stage game}
We consider an auction market with a set of bidders $N=\{1,2,\cdots,n\}$ and a set of sellers/exchanges $\{1,2,\cdots,m\}$. Each bidder $i$ has a value $v_i\geq 0$ for the item sold by the exchanges in the market. In this paper we focus on symmetric setting, where each bidder's valuation $v_i$ is drawn independently and identically from a publicly known distribution $F$. Throughout the paper, we assume the value distribution $F$ is differentiable on $[0,\infty)$, have no point mass and the continuous density function $f$ is positive on $(0,\infty)$. 

While bidders repeatedly interact with the exchanges billions of times a day, every such interaction is a priori identical: namely, with a probability $\lambda_j$, a query arrives from exchange $j$ (here $\lambda_j$ indicates exchange $j$'s market power) asking the bidder to submit a bid in exchange $j$'s mechanism $\mech_j$, and in each such interaction every bidder's value is drawn independently from $F$ (independently of previous rounds and other bidders). A bidder $i$'s bid in every individual interaction is a function exclusively of his own value $v_i$ and other bidders' distribution $F$, and in particular \emph{not a function of which specific exchange} $j$ sent this query (see the Introduction for motivation). Given the a priori identical interactions, it follows that it is sufficient to model and study a single such interaction, which is what we do in this paper.

We model the whole market as a two-stage game. The game is parametrized by a class of mechanisms $\C$:
\begin{itemize}
\item In the first stage, each exchange $j$ proposes a mechanism $\mech_j$ simultaneously. The mechanism $\mech_j=(\alloc^j,\pay^j) \in \C$ is a sealed-bid auction that sells a single item to $n$ bidders. For any bidding profile $\bids=(b_1,b_2,\cdots,b_n)\in\R^n$ submitted by the bidders, the auction $\mech_j$ charges each bidder $i$ a non-negative price $p^j_i(\bids)$, and allocates the item to bidder $i$ with probability $x^j_i(\bids)$. 
\item In the second stage, each bidder $i$ proposes a bid $b_i=\bidf_i(v)$ according to his private value $v_i$, here $\bidf_i:\R\to\R$ is the bidding function of bidder $i$. The same bid will be submitted regardless of the exchange $j$ that solicits bids this query. Given that each query comes with probability $\flow_j$ from exchange $j$, the bidder will get allocation $x_i(\bids)=\sum_{j}\flow_j x^j_i(\bids)$ and charged $p_i(\bids)=\sum_{j}\flow_j p^j_i(\bids)$. In other words, the bidders respond to the average mechanism $(\alloc,\pay)=(\sum_{j}\flow_j\alloc_j,\sum_{j}\flow_j\pay_j)$ proposed by the exchanges (again, see Introduction for motivation). Later in the paper we may also refer to such auction as $\sum_{j=1}^{m}\flow_j\mech_j$.
\end{itemize}
In this two-round game, the objective of each bidder $i$ is to maximize his expected utility. Since the setting is symmetric, we are concerned with symmetric bidding equilibrium only, where the bidders are using the same bidding function $\bidf=\bidf_1=\bidf_2=\cdots=\bidf_m$ in the Bayes-Nash equilibrium. In the classes $\C$ of mechanisms that we analyze, a symmetric bidding equilibrium always exists. For the class of mechanisms $\C$ considered in Section~\ref{sec:withreserve}, the symmetric bidding equilibrium is essentially unique\footnote{I.e., unique if we ignore the bids below the smallest reserve price (it's okay to ignore what happens in this regime because any such bid leads to no allocation and generates zero payment), and if we ignore the bids at points of discontinuity of the bidding function, where the bidder is indifferent between the two points of discontinuity.}. For the class $\C$ considered in Section~\ref{sec:noreserve}, every symmetric bidding equilibrium implements the same social choice function (i.e., the allocation) and charges identical interim payments, thereby making uniqueness of symmetric bidding equilibrium irrelevant. The objective of each exchange $j$ is to maximize its expected revenue when bidders are bidding against the average mechanism, which is $$\rev_j = \E_v\left[\sum_{i}p^j_i(\bidfs(\values))\right].$$

As usual in mechanism design, when choosing a new auction rule the exchanges need to reason about how it will affect the bidding strategies of the buyers. The twist in this model is that the chosen mechanism only partially affects the average mechanism that the buyers are reacting to. This leads to a situation very much like \emph{tragedy of the commons} games.
\section{Warm Up: High-bid auctions}
\label{sec:noreserve}
As a warm up we start by analyzing the game between exchanges where each exchange is restricted to using a \emph{high-bid auction}, i.e., an auction where the item is always allocated to the highest bidder. This is satisfied by various natural auction formats: second price, first-price, pricing at any convex combination of first and second price, ... This apparently innocuous assumption, however, prevents the use of reserve prices -- which is the main lever used in practice for revenue optimization. We will come back to reserve prices in the next section. 


Analyzing the equilibrium of the market seems challenging, as the bidders' bidding strategy changes dramatically whenever an exchange proposes a different auction. Surprisingly, if each exchange proposes an auction from above class of mechanisms, we can show that each exchange proposing first-price auction is the unique Nash equilibrium.

\begin{theorem}
\label{thm:noreserve}
If every exchange is only allowed to use auctions with the following properties:
\begin{enumerate}
    \item Highest bidder always wins (this implies a zero reserve price);
    \item Winner pays no more than his bid;
    \item Losers always pay 0;
    \item The bidders have a pure-strategy symmetric Bayes-Nash bidding equilibrium.
\end{enumerate}
Then each bidder proposing first-price auction is a unique Nash equilibrium.
\end{theorem}

\begin{proof}
Consider the aggregated average mechanism. Since highest-valued bidder always wins in any auction proposed by any exchange, thus highest-valued bidder always wins in aggregated average mechanism. Notice that Revenue Equivalence Theorem (see \citet{vickrey1961counterspeculation,Riley1981optimal,myerson1981optimal}) implies that auctions whose equilibria implement the same social choice function always extract the same revenue in those equilibria. Given that bidders have a pure-strategy symmetric BNE, and the highest bidder always wins, we have that the highest valued buyer always wins --- i.e., the same social choice function. Thus the expected sum of revenue over all exchanges $\sum_i \lambda_i \rev_i$ will be a constant, no matter what auctions are proposed by the exchanges (i.e., any auction satisfying properties (1) and (4) in the theorem statement, which get used in the Revenue Equivalence Theorem).

Fixing the bids of the bidders, and respecting properties (2) and (3) in the theorem statement, a first-price auction extracts the highest possible revenue. An exchange running first-price auction will extract at least average revenue among all exchanges, and the equality holds only if all exchanges are running first-price auction. If all exchanges are running first-price auction, then no one would deviate to other auction to get below-average revenue.

If not all exchanges are running first-price auction, there are two possible cases: (i) either some but not all exchanges are using first-price auctions. In such case some exchange $j$ must be getting below-average revenue. By deviating to first-price, this exchange can guarantee at least the average revenue since $\rev_j \geq \rev_i$ and hence $\rev_j \geq \sum_i \lambda_i \rev_i$ which is constant. (ii) if no exchange is using first-price, then there is some exchange $j$ that is at the average or below. By switching to first-price, this exchange can guarantee that  $\rev_j > \rev_i$ for all $i \neq j$ and hence $\rev_j > \sum_i \lambda_i \rev_i$.

In either case, if not all exchanges are running first-price, there is at least one exchange that would switch to first-price auction to get strict revenue improvement. Thus in a Nash equilibrium, every exchange must propose the first-price auction.
\end{proof}

\paragraph{Necessity of properties 2 and 3:} The second and the third properties in Theorem~\ref{thm:noreserve} cannot be removed. For the second property, suppose that every exchange proposes first-price auction. If an exchange switches to the following auction: highest bidder wins and pays 100 times his bid. Then that exchange will get more revenue than all other exchanges using first-price auction, thus above-average revenue. This means that his revenue increases after making such change to proposed auction. For the third assumption, again suppose that every exchange uses first-price auction. Then an exchange would prefer to switch to all-pay auction, where each bidder pays his bid. This leads to his revenue being above average among all exchanges, thus increases his revenue. 

We argue that the second and the third properties are also natural in real-world auctions. However, the first property is a huge restriction to the class of possible mechanisms. For example, most ad exchanges use first-price auction or second-price auction with reserve prices. An auction with reserve price violates the first property, since if the bids of all bidders are below the reserve price, no bidder will get served. In the next section, we discuss what happens when exchanges are allowed to use auctions with reserve prices.

As for property 4, most natural auctions have a pure strategy symmetric BNE. For example, first-price auction, second price auction, a highest-bidder-wins and pays a convex combination of first and second highest bid, and even the all-pay auction (that violates property (3)), all satisfy property (4).
\section{Auctions with reserve price}
\label{sec:withreserve}

\subsection{Main results}
\label{sec:mainresults}
When the highest bidder does not always get the item, the convenient and straightforward application of revenue equivalence theorem becomes out of question\footnote{Although the revenue equivalence theorem itself remains true, any straightforward application of it is not possible.}. Even in very simple settings, as we demonstrate in Section~\ref{sec:simpleex}, the details are quite involved. 

The most ubiquitous auctions used in real-world ad exchanges are first-price auction with reserve prices and second-price auction with reserve prices. We show that given these choices, no exchange will propose second-price auction with reserve in an equilibrium.

\begin{theorem}
\label{thm-fpspne}
In a market with $m\geq 2$ exchanges, if each exchange is only allowed to use first-price auction with reserve or second-price auction with reserve, then in a pure-strategy equilibrium, every exchange will use first-price auction with reserve.
\end{theorem}

Now that we have established every exchange would pick a first-price auction with a reserve, we study properties of exchanges' equilibrium in such a first-price market. The first property is that every exchange using first-price auction \emph{without reserve} cannot be an equilibrium of the market. Note that this immediately justifies the necessity of property 1 in Theorem \ref{thm:noreserve}.
\begin{theorem}\label{thm-allzero}
Every exchange proposing first-price auction with no reserve is not a Nash equilibrium.
\end{theorem}
The second property is that the market does not admit symmetric pure-strategy equilibrium.
\begin{theorem}\label{thm-asymmetric}
Every exchange proposing first-price auction with the same reserve $r$ is not a Nash equilibrium.
\end{theorem}
A direct corollary of the above two theorems is that the competition in the market leads to the decrease of total revenue of all exchanges. We know that for each bidder having regular value distribution $F$, second-price auction with reserve is the revenue-optimal auction \cite{myerson1981optimal}, and every revenue-optimal auction must have the same allocation rule\footnote{Indeed the revenue-optimal auction is unique for strictly regular distributions because Myerson's theorem~\cite{myerson1981optimal} shows that the optimal auction should not have a non-zero probability of allocation to any agent with a negative virtual value or to any agent with lower than highest virtual value.}. By revenue equivalence theorem, first-price auction with the same reserve is also revenue-optimal. However, since every exchange proposing first-price auction with the same reserve cannot be an equilibrium, it means that in a competitive market the total revenue of exchanges decreases. To summarize, we have the following corollary.
\begin{corollary}\label{cor-lessrev}
Suppose the bidders' valuation distribution $F$ is strictly regular (virtual values are strictly increasing). Then in an equilibrium of the market, the total revenue of all exchanges is lower than in the revenue-optimal auction.
\end{corollary}
In the next section, we will study the above theorems in a specific market setting, and prove the theorems in general case in section \ref{sec:thmproofs}.

\subsection{An example study with two exchanges and two bidders, in a first-price market}
\label{sec:simpleex}
Let's study a simple example with two bidders and two exchanges, and study the equilibrium of the market. Suppose that there are two bidders with values drawn from uniform distribution $F=U[0,1]$. There are two exchanges, each with $\flow_1=\flow_2=0.5$ fraction of traffic. Let $\rev_j(\mech_1,\mech_2)$ denote the revenue obtained by exchange $j$ when the two exchanges propose mechanisms $\mech_1$ and $\mech_2$. Let $\fp_r$ denote first-price auction with reserve $r$. Suppose we only allow each exchange to propose first-price auction, but with an arbitrary reserve of its choosing. Assume that exchange 1 proposes $\fp_{r_1}$, while exchange 2 proposes $\fp_{r_2}$. Since the setting is symmetric, without loss of generality assume that $r_1\leq r_2$. 

\paragraph{Straightforward revenue equivalence is fruitless:} First we note that the sum of revenue of the two exchanges when they run $\fp_{r_1}$ and $\fp_{r_2}$ is \emph{not revenue equivalent} to the setting where both the exchanges run second price auctions with reserves of $r_1$ and $r_2$, i.e., when they run $\secp_{r_1}$ and $\secp_{r_2}$. To see this consider a setting where $r_1 = 0.7 \leq r_2 = 0.9$.

When the exchanges ran $\secp_{r_1}$ and $\secp_{r_2}$ the bidders bid truthfully. Suppose $v_2 = r_2 + \epsilon = 0.9 + \epsilon$ for $\epsilon > 0$. Bidder $2$ will win the item with probability $0.9 + \epsilon$ (because bidder $1$'s value is drawn from $U[0,1]$).

On the other hand, in the first-price auction, bidder $2$ will shade his bids (i.e., not bid truthfully) and therefore will not get allocated with probability more than $0.5$, even when his value is $v_2 = 0.9 + \epsilon$.
Indeed, suppose bidder $2$ was allocated when exchange $2$ was running the auction\footnote{it is enough to consider deterministic bidding strategies for bidder $2$ as his utility would be equal in every single best response he has, and clearly there is always a deterministic best response}. Then clearly $b_2$ should have been at least $r_2$, making $2$'s payment (which is just $b_2$) also at least $r_2$, and his utility is at most $\epsilon$. On the other hand, suppose bidder $2$ just set $b_2 = 0.7$. Since $b_1 \leq v_1$, bidder $2$ will be served with probability at least $0.7 * 0.5$ (here $0.5$ captures $\lambda_1$ and $0.7$ captures $U[0,1]$ distribution) and therefore get a utility of at least $0.5 * 0.7 * (0.9+\epsilon-0.7) = 0.07 + 0.35\epsilon > \epsilon$ for sufficiently small $\epsilon$. Thus, values slightly higher than $r_2$ would bid lower than $r_2$, leading to a different allocation than the second price auction, thus denying the possibility of applying revenue equivalence easily.

\paragraph{Discontinuous bidding functions:} Another peculiarity that manifests is the discontinuity in the equilibrium bidding function, even when the distribution is nice like $U[0,1]$. To see this, let $\beta(\cdot)$ be the equilibrium bidding function (which is clearly weakly monotone). Consider the smallest value $v$ s.t. $\beta(v) = r_2$. Under any best response bidding function, the utility of buyer with value $v-\epsilon$ is at most $\epsilon$ smaller than the utility of buyer with $v$ (buyer with $v-\epsilon$ can simply imitate $v$'s bid to satisfy this). But a bidder with $v > r_2$ bidding at $b \geq r_2$ gets allocated with about twice the probability and gets about twice the utility when compared to the buyer with value infinitesimally smaller than $v$ bidding $b < r_2$. This is a contradiction, thereby ruling out continuity of the bidding function.

Let's first compute the equilibrium bidding strategy of the bidders and the revenue obtained by each exchange.

\begin{proposition}
Let $a=\frac{2r_2+\sqrt{4r_2^2-3r_1^2}}{3}$ and $c=\frac{4r_2^2+3r_1^2+2\sqrt{4r_2^2-3r_1^2}}{9}$. Then 
\begin{itemize}
    \item Under equilibrium bidding, both bidders use bidding function $$\bidf(v)=\frac{v^2+r_1^2}{2v} \text{ for } v\in[r_1,a)$$ $$\bidf(v)=\frac{v^2+c}{2v} \text{ for } v\in[a,1].$$
    \item The revenue is described in the following formulas:
    $$\rev_1(\fp_{r_1},\fp_{r_2})=-\frac{4}{3}r_1^3+r_1^2a+c-ca+\frac{1}{3}$$
    $$\rev_2(\fp_{r_1},\fp_{r_2})=\frac{1}{3}-\frac{1}{3}a^3+c-ca.$$
\end{itemize}
\end{proposition}
\begin{proof}
The bidding function $\bidf$ has the following properties. When $v<r_1$, $\bidf(v)<r_1$, as the winner has to pay at least $r_1$ when he wins; therefore a bid of $0$ is as good for the bidder as any other bid below $r_1$. There exists some value $a$ such that $\rbidf(a)=r_2$, and $\lbidf(a)<r_2$, here $\rbidf$ and $\lbidf$ denote the right-side limit and left-side limit of the bidding function $\bidf$. 

Since $\bidf$ is a symmetric equilibrium, the bidder with value $v$ will bid $b=\bidf(v)$ that maximizes the expected utility, which is $F(\bidf^{-1}(b))(v-b)$, i.e. the probability of winning multiplied by the utility when he wins. Let $f=F'$ be the density function of $F$. Taking the first-order condition we have 
\[0=\frac{f(\bidf^{-1}(b))}{\bidf'(\bidf^{-1}(b))}(v-b)-F(\bidf^{-1}(b)).\]
Apply $b=\bidf(v)$ to above equation we get
\begin{equation}\label{eqn-first-order}
f(v)v=\bidf'(v)F(v)+f(v)\bidf(v)=\frac{d}{dv}F(v)\bidf(v).
\end{equation}
Integrate both sides for $v\in(r_1,a)$ we have
\begin{equation}\label{eqn-first-part}
\int_{r_1}^{v}tf(t)dt=F(v)\lbidf(v)-F(r_1)\rbidf(r_1).
\end{equation}
Apply $F(v)=v$ to above equation we get $\bidf(v)=\frac{v^2+r_1^2}{2v}$ for $v\in[r_1,a)$.
Similarly integrate both sides of equation (\ref{eqn-first-order}) for $v\in[a,1]$ and apply $F(v)=v$ we get $\bidf(v)=\frac{v^2+c}{2v}$ for $v\in[a,1]$, here $c=2ar_2-a^2$ is a constant.

\begin{figure}[htbp]
        \centering
        \includegraphics[width=0.4\doubletextwidth]{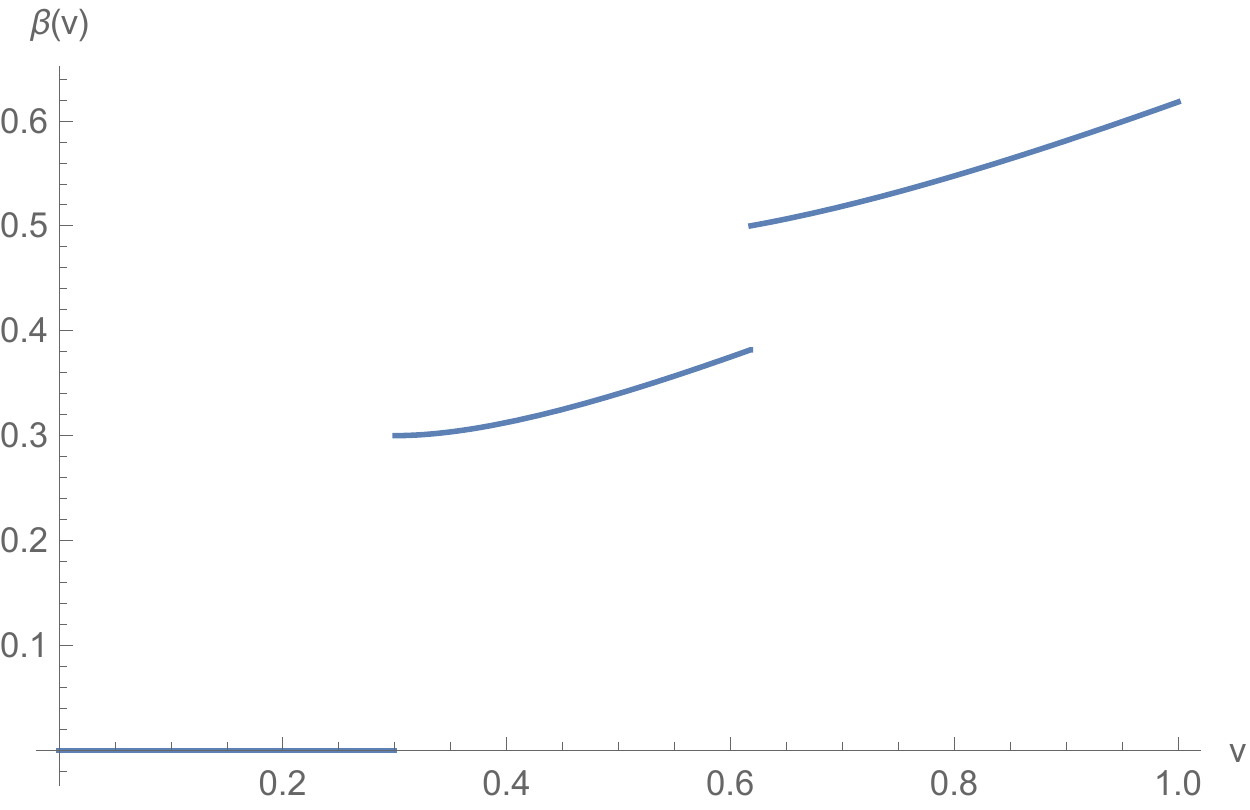}
        \caption{An example of equilibrium bidding function when the exchanges propose auctions $\fp_{0.3}$ and $\fp_{0.5})$. One can notice the discontinuity of such bidding function.}
        \label{fig-example}
\end{figure}

Now we solve the value of $a$. The bidder will be indifferent to bid $\rbidf(a)$ and $\lbidf(a)$ at value $a$, thus the utility he gets would be the same as these two bids, in other words, $a-\lbidf(a)=2(a-\rbidf)=2(a-r_2)$. Apply $\lbidf(a)=\frac{a^2+r_1^2}{2a}$ we get $a=\frac{2r_2+\sqrt{4r_2^2-3r_1^2}}{3}$ and $c=\frac{4r_2^2+3r_1^2+2\sqrt{4r_2^2-3r_1^2}}{9}$. 

Notice that the distribution of the higher value of two bidders has cumulative density function $v^2$ and density function $2v$. Thus the revenue of exchange 1 would be 
\begin{eqnarray*}
\rev_1(\fp_{r_1},\fp_{r_2})&= &\int_{r_1}^{a}2v\bidf(v)dv+\int_{a}^{1}2v\bidf(v)dv\\
&=&\int_{r_1}^{a}(v^2+r_1^2)dv+\int_{a}^{1}(v^2+c)dv\\
&=&-\frac{4}{3}r_1^3+r_1^2a+c-ca+\frac{1}{3}.
\end{eqnarray*}
The revenue of exchange 2 is 
\[\rev_2(\fp_{r_1},\fp_{r_2})=\int_{a}^{1}2v\bidf(v)dv=\frac{1}{3}-\frac{1}{3}a^3+c-ca.\]
\end{proof}

Given the characterization of bidding function and revenue function in the market, we are ready to verify the theorems and corollaries in previous section.
\begin{corollary}\label{cor-reserve-useful} (Special case of Theorem \ref{thm-allzero})
If one exchange proposes first-price auction with zero reserve, the other exchange may propose first-price auction with some reserve price to gain more revenue.
\end{corollary}

\begin{proof}
When both exchanges use no reserve, the revenue of exchange 1 is $\frac{1}{3}$. However, if he switches to first-price auction with reserve 0.1, his revenue would be $\rev_1(\fp_{0.1},\fp_0)\approx0.3402>\frac{1}{3}$. Thus setting up a reserve may lead to higher revenue for the exchange.
\end{proof}

\begin{corollary}\label{cor-asymmetric-reserve}(Special case of Theorem \ref{thm-asymmetric})
The pure-strategy equilibrium of the market will not be symmetric. That is, both exchanges proposing first-price auction with the same reserve cannot be Nash equilibrium.
\end{corollary}

\begin{proof}
Corollary \ref{cor-reserve-useful} shows that $(\fp_0,\fp_0)$ cannot be an equilibrium. It remains to show that $(\fp_r,\fp_r)$ cannot be an equilibrium for any $r\in(0,1]$. The idea is for any $r\in(0,1]$, there exists small $\epsilon>0$ such that $\rev_1(\fp_{r-\epsilon},\fp{r})>\rev_1(\fp_{r},\fp_{r})$. This can be done by verifying $\frac{\partial}{\partial r_1}\rev_1(\fp_{r_1},\fp_{r_2})\Big|_{r_1=r_2}=-2r_2^2\leq0$ for any $r_2\in[0,1]$. 
\end{proof}

\begin{corollary}\label{cor-speciallessrev}
(Special case of Corollary \ref{cor-lessrev}) Pure-strategy equilibrium of exchanges in the market will not be revenue-optimal. 
\end{corollary}

\begin{proof}
Given the revenue function of the exchanges, we can verify that reserve prices $r_1=0.2402$, $r_2=0.3157$ corresponds to the unique equilibrium of the market. The revenue of exchange 1 is 0.397, while the revenue of exchange 2 is 0.378. The total revenue of the market under equilibrium is $\frac{1}{2}*0.397+\frac{1}{2}*0.378=0.3875$. However, in revenue-optimal auction (i.e. first-price auction with reserve $0.5$), the revenue would be $\frac{5}{12}=0.4167>0.3875$.
\end{proof}

Although the revenue decreases in equilibrium of the market, social welfare (i.e. the expected value obtained by the bidders) increases.

\begin{corollary}
Pure-strategy equilibrium of the market generates more welfare than revenue-optimal auction.
\end{corollary}
\begin{proof}
Reserve prices $r_1=0.2402$, $r_2=0.3157$ corresponds to the unique equilibrium of the market. We can compute the lowest winner's value when both exchanges sell the item: $a=\frac{2r_2+\sqrt{4r_2^2-3r_1^2}}{3}=0.369$. Since in revenue-optimal auction, the auction only sells the item when winner's value is at least 0.5. Thus the welfare in such first-price market's equilibrium is higher than in revenue-optimal auction. 
\end{proof}

\subsection{Proof of theorems in Section \ref{sec:mainresults}}
\label{sec:thmproofs}

\begin{proof}[Proof of Theorem \ref{thm-fpspne}.]
The proof of the theorem is structured as follows. Firstly we characterize the bidding equilibrium of bidders in a market with first-price and second-price auctions. We show the equilibrium bidding functions have discontinuous segments, and analyze the equation the bidding function need to satisfy in each segment. Secondly we prove that for any exchange with reserve $r_{\ell}$, if all exchanges with smaller reserves use first-price auctions, then such exchange will prefer to propose first-price auction with reserve $r_{\ell}$ rather than second-price auction with the same reserve $r_{\ell}$. We prove this by showing that conditioned on highest-valued bidder has any value $v$, the expected payment received by such exchange when proposing $\fp_{r_{\ell}}$ will always be as large as proposing $\secp_{r_{\ell}}$. This will imply that no exchange in an equilibrium will prefer to propose second-price auction with reserve.

\textit{\textbf{Step 1: Characterize the symmetric bidding equilibrium of bidders.}} 
Given a bidder with value $v$, let $G(v)=F^{n-1}(v)$ be the probability that the second highest value is smaller than $v$, and $g(v)=G'(v)$ be the corresponding density function. Assume that exchanges proposing $\mech^1_{r_1},\mech^2_{r_2},\cdots,\mech^m_{r_m}$ forms an auction equilibrium, here each auction $\mech^i\in\{\fp,\secp\}$ can only be first-price or second price, and without loss of generality $r_1\leq r_2\leq\cdots\leq r_m$. To prove the theorem, it suffices to show that if in an equilibrium the auctions with lower reserves are first-price auctions, i.e. $\mech^1=\mech^2=\cdots=\mech^{k-1}=\fp$ for some $k\in[m]$, then $\mech^{k}=\fp$.

We first investigate what does the symmetric equilibrium bidding function $\bidf$ of the bidders look like. Similar to the case with two exchanges that we have studied before, the equilibrium bidding function would have at most $m+1$ continuous segments (also see Figure \ref{fig-example}). To be more precise, there exists $a_1\leq a_2\leq\cdots\leq a_m$ such that $\rbidf(a_\ell)=r_\ell$; while $\lbidf(a_\ell)<r_\ell$, if $r_{\ell-1}<r_\ell$. Let $a_0=0$ and $a_{m+1}=+\infty$. For $v\in[a_\ell,a_{\ell+1})$, let $\avgp_\ell(b)$ denote the average payment of the winner in aggregated average auction of $\mech^1_{r_1},\mech^2_{r_2},\cdots,\mech^\ell_{r_\ell}$ (which is $\frac{1}{\sum_{j=1}^{\ell}\flow_j}\sum_{j=1}^{\ell}\flow_j\mech^j_{r_j}$), conditioned on winner's bid being $b$.

When highest-valued bidder has value $a_\ell$, he will bid $\lbidf(a_\ell)$ and $\rbidf(a_\ell)$ indifferently. When bidding $\lbidf(a_\ell)$, he wins in the first $\ell-1$ exchanges, and gets profit $\left(\sum_{j=1}^{\ell-1}\flow_j\right)\Big(a_\ell-\avgp_{\ell-1}(\lbidf(a_\ell))\Big)$. When bidding $\rbidf(a_\ell)$, he wins in the first $\ell$ exchanges, and gets profit $\left(\sum_{j=1}^{\ell}\flow_j\right)\Big(a_\ell-\avgp_{\ell}(\rbidf(a_\ell))\Big)$. Then the value of $a_\ell$ will be determined by the following equation:
\begin{equation}\label{eqn-al}
\left(\sum_{j=1}^{\ell-1}\flow_j\right)\Big(a_\ell-\avgp_{\ell-1}(\lbidf(a_\ell))\Big)=\left(\sum_{j=1}^{\ell}\flow_j\right)\Big(a_\ell-\avgp_{\ell}(\rbidf(a_\ell))\Big).
\end{equation}
Now we analyze first-order condition similar to what we did in previous section. Since $\bidf$ is a bidding equilibrium, when winner's value $v\in [a_\ell,a_{\ell+1})$, he bids $b=\bidf(v)$ to maximize his utility $G(\bidf^{-1}(b))\left(\sum_{j=1}^{\ell}\flow_j\right)(v-\avgp_\ell(b))$. Take the derivative to zero, we get
\[\frac{g(\bidf^{-1}(b))}{\bidf'(\bidf^{-1}(b))}(v-\avgp_\ell(b))-G(\bidf^{-1}(b))p'_\ell(b)=0.\]
Applying $b=\bidf(v)$ we get
\[vg(v)=\bidf'(v)G(v)\avgp_\ell(\bidf(v))+g(v)p'_\ell(\bidf(v))=\frac{d}{dv}G(v)\avgp_\ell(\bidf(v)).\]
Integrate both sides from $a_\ell$ to $v$ we get
\begin{equation}\label{eqn-integral-ell}
\int_{a_\ell}^{v}g(t)tdt=G(v)\avgp_\ell(\lbidf(v))-G(a_\ell)\avgp_\ell(\rbidf(a_\ell)).
\end{equation}
The following lemma ensures that there is a unique equilibrium bidding function $\bidf$ that satisfy (\ref{eqn-integral-ell}) for every $\ell$. The proof of the lemma will get deferred to appendix.

\begin{lemma}\label{lem-existence}
There exists a solution bidding function $\bidf$ to (\ref{eqn-integral-ell}) for all $\ell$, and such bidding function $\bidf$ is an equilibrium among bidders. Furthermore, $\bidf(v)\leq v$ for all $v$.
\end{lemma}

\textit{\textbf{Step 2: Find a lower bound for exchange $k$'s revenue when he proposes $\fp_{r_k}$.}}
Now we are ready to analyze the revenue obtained by exchange $k$, when all exchanges with smaller reserve prices use first-price auctions. We start with an upper bound of the revenue of exchange $k$ when he proposes $\fp_{r_k}$. Consider any winner's value $v\geq a_k$, and assume $v\in[a_{s},a_{s+1})$ for some $s$. Then
\begin{eqnarray*}
\int_{a_k}^{v}tg(t)dt
&=&\int_{a_s}^{v}tg(t)dt+\int_{a_{s-1}}^{a_s}tg(t)dt+\cdots+\int_{a_k}^{a_{k+1}}tg(t)dt\\
&=&\left(G(v)\avgp_s(\lbidf(v))-G(a_s)\avgp_s(\rbidf(a_s))\right)\\
& &+\left(G(a_{s})\avgp_{s-1}(\lbidf(a_s))-G(a_{s-1})\avgp_{s-1}(\rbidf(a_{s-1}))\right)\\
& &+\cdots+\left(G(a_{k+1})\avgp_k(\lbidf(a_{k+1}))-G(a_k)\avgp_k(\rbidf(a_k))\right)\\
&\leq&G(v)\avgp_{s}(\lbidf(v))-G(a_k)\avgp_k(\rbidf(a_k))\\
&\leq&G(v)\lbidf(v)-G(a_k)\avgp_k(\rbidf(a_k))\\
&=&G(v)\lbidf(v)-G(a_k)r_k.
\end{eqnarray*}
Here the second equality follows from equation (\ref{eqn-integral-ell}); the third inequality is true since $\avgp_{\ell-1}(\lbidf(a_\ell))<\avgp_{\ell}(\rbidf(a_\ell))$ by equation 
(\ref{eqn-al}), whenever $r_{\ell-1}<r_{\ell}$; the fourth inequality is by $\avgp_s(\lbidf(v))\leq \lbidf(v)$ as all auctions proposed by exchanges are first-price with reserve or second-price with reserve which cannot charge the winner more than his bid, thus the aggregated average auction also won't charge the winner more than his bid, and the equality holds only if every exchange with reserve at most $r_s$ uses first-price auction; the last equality follows from the fact that $\rbidf(a_k)=r_k$, and all exchanges with smaller reserves propose first-price auctions, thus $\avgp_k$ is identity function. Thus when exchange $k$ uses first-price auction with reserve $r_k$, when highest-valued bidder has value $v\geq a_k$, the revenue would be $\lbidf(v)$\footnote{Here we need to assume $\lbidf(v)=\bidf(v)$. This is not true for $v=a_\ell$ for some $\ell$, but we can ignore the revenue contribution from these points, as we assume the value distribution function has no point mass.}, which is lower bounded by 
\begin{equation}\label{eqn-fplb}
\lbidf(v)\geq\frac{1}{G(v)}\left(G(a_k)r_k+\int_{a_k}^{v}tg(t)dt\right). 
\end{equation}
For large enough $v$, the equality only holds if all exchanges use first-price auction.

\textit{\textbf{Step 3: Find an upper bound for exchange $k$'s revenue when he proposes $\secp_{r_k}$.}}
Now we upper bound the revenue of exchange $k$ when he uses second-price auction with reserve $r_k$, when the highest-valued bidder has value $v\geq a_k$. Let $v^{(1)}=v$ be the highest value among the bidders, $v^{(2)}$ be the second highest value. Notice that the revenue of exchange $k$ when $v^{(1)}=v$ and $v^{(2)}>a_k$ can be upper bounded by the second highest value since the bidders never overbid in the aggregated auction by Lemma \ref{lem-existence}. Thus
\begin{eqnarray}
\E_{\values}[\max(\bidf(v^{(2)}),r_k)|v^{(1)}=v]
&=&\Pr[v^{(2)}<a_k|v^{(1)}=v]r_k\nonumber\\
& &+\Pr[v^{(2)}\geq a_k|v^{(1)}=v]\cdot\E[\bidf(v^{(2)})|v^{(2)}\geq a_k,v^{(1)}=v]\nonumber\\
&=&\frac{G(a_k)}{G(v)}r_k+\frac{G(v)-G(a_k)}{G(v)}\cdot\frac{\int_{a_k}^{v}\bidf(t)d\frac{G(t)}{G(v)}}{\int_{a_k}^{v}d\frac{G(t)}{G(v)}}\nonumber\\
&=&\frac{1}{G(v)}\left(G(a_k)r_k+\int_{a_k}^{v}\bidf(t)g(t)dt\right)\nonumber\\
&\leq&\frac{1}{G(v)}\left(G(a_k)r_k+\int_{a_k}^{v}tg(t)dt\right)\label{eqn-second-bid}.
\end{eqnarray}
The last inequality follows from $\bidf(t)\leq t$. For large enough $v$, the equality holds only if every exchange proposes second-price auction, which means the equality of (\ref{eqn-second-bid}) cannot hold simultaneously with (\ref{eqn-fplb}), unless there is only one exchange.
Thus for every possible winner's value $v$, proposing $\secp_{r_k}$ always yields less revenue than proposing $\fp_{r_k}$. Taking the expectation over all possible values of $v^{(1)}$, we show that proposing $\fp_{r_k}$ is always a strictly better strategy than proposing $\secp_{r_k}$. Thus in an equilibrium of market, no exchange will propose second-price auction with reserve.
\end{proof}

\begin{proof}[Proof of Theorem \ref{thm-allzero}]
By definition of Nash Equilibrium, we need to prove that an exchange in market $(\fp_0,\fp_0,\cdots,\fp_0)$ will deviate to positive reserve price to gain revenue. Since the other exchanges have the same reserve prices, we can merge them and treat them as a single exchange. It suffices to prove when two exchanges propose $(\fp_0,\fp_0)$, there exists $\epsilon>0$ such that the first exchange can deviate to $\fp_{\epsilon}$ to get more revenue.

By equation (\ref{eqn-integral-ell}), when both exchanges use $\fp_{0}$, the bidding function is 
\begin{equation*}
\bidfn(v)=\frac{\int_{0}^{v}tg(t)dt}{G(v)},\ v\geq 0.
\end{equation*}
When exchange 1 switches to $\fp_{\epsilon}$, the bidding function $\bidf$ of the bidder will be as follows:
\begin{equation*}
\bidf(v)=\begin{cases}
    \frac{\int_{0}^{v}tg(t)dt}{G(v)}, &0\leq v<a;\\
    \frac{\epsilon G(a)+\int_{a}^{v}tg(t)dt}{G(v)}, &v\geq a.
\end{cases}
\end{equation*}
Here $a$ is the discontinuous point where bidder is indifferent in bidding $\lbidf(a)$ and $\rbidf(a)$, which is determined by $a-\epsilon=\flow_1(a-\lbidf(a))$. This implies $\int_{0}^atg(t)dt=\frac{\epsilon-(1-\flow_1)a}{\flow_1}G(a)$, and $\epsilon=(1-\flow_1)a+\flow_1\frac{\int_{0}^atg(t)dt}{G(a)}$.
When exchange 1 switches from $\fp_0$ to $\fp_{\epsilon}$, he will gain revenue when $v^{(1)}\geq a$, and lose revenue when $v^{(1)}<a$. Conditioned on $v=v^{(1)}\in[a,\infty)$, the revenue gain is 
\begin{eqnarray*}
\bidf(v)-\bidfn(v)
&=&\frac{\epsilon G(a)+\int_{a}^{v}tg(t)dt}{G(v)}-\frac{\int_{0}^{v}tg(t)dt}{G(v)}\\
&=&\frac{\epsilon G(a)-\int_{0}^{a}tg(t)dt}{G(v)}\\
&=&\frac{(1-\flow_1)}{G(v)}\left(aG(a)-\int_{0}^atg(t)dt\right).
\end{eqnarray*}
Here the last equality is by $\epsilon=(1-\flow_1)a+\flow_1\frac{\int_{0}^atg(t)dt}{G(a)}$.
Conditioned on $v=v^{(1)}\in[0,a)$, the revenue loss is 
\[\bidfn(v)-\bidf(v)=\bidfn(v)=\frac{\int_{0}^{v}tg(t)dt}{G(v)}.\]
Take $\epsilon$ small enough such that $F(a)\leq\frac{n-1}{n}$. Recall that $G(v)=F^{n-1}(v)$. Since winner's value $v^{(1)}$ has cumulative density function $F^n(v)$ and density $nF^{n-1}(v)f(v)=nG(v)f(v)$, the total revenue gain is at least
\begin{eqnarray*}
\int_{0}^{\infty}nf(t)G(t)(\bidf(t)-\bidfn(t))
&=&\int_{a}^{\infty}nf(t)G(t)\frac{(1-\flow_1)}{G(t)}\left(aG(a)-\int_{0}^asg(s)ds\right)dt\\
& &-\int_{0}^{a}nf(t)G(t)\frac{\int_{0}^{t}sg(s)ds}{G(t)}dt\\
&=&n(1-F(a))(1-\flow_1)\left(aG(a)-\int_{0}^atg(t)dt\right)\\
& &-\int_{0}^{a}nf(t)\int_{0}^{t}sg(s)dsdt\\
&\geq&(1-\flow_1)\left(aG(a)-\int_{0}^atg(t)dt\right)-\int_{0}^{a}nf(t)\int_{0}^{t}sg(s)dsdt.
\end{eqnarray*}
Here the last inequality is by $F(a)\leq \frac{n-1}{n}$. Define \[h(a)=(1-\flow_1)\left(aG(a)-\int_{0}^atg(t)dt\right)-\int_{0}^{a}nf(t)\int_{0}^{t}sg(s)dsdt,\] and easy to see that $h(0)=0$. To prove that exchange 1's revenue gain is positive when switching to $\fp_{\epsilon}$ for some small $\epsilon>0$, it suffices to show that $h'(a)>0$, for every small enough $0<a<\frac{1-\flow_1}{nc}$. Actually,
\begin{equation*}
h'(a)=(1-\flow_1)G(a)-nf(a)\int_{0}^{a}tg(t)dt.
\end{equation*}
Find constant $c$ such that $f(a)\leq c$ for small enough $a$. Then 
\begin{eqnarray*}
h'(a)&\geq&(1-\flow_1)\int_{0}^{a}g(t)dt-nc\int_{0}^{a}tg(t)dt\\
&=&\int_{0}^{a}\left(1-\flow_1-nct\right)g(t)dt>0.
\end{eqnarray*}
Thus there exists $a>0$ such that $h(a)>0$, which implies exchange 1 gains revenue by switching to $\fp_{\epsilon}$. Therefore $(\fp_0,\fp_0)$ cannot be an equilibrium.
\end{proof}

\begin{proof}[Proof of Theorem \ref{thm-asymmetric}]
We prove that an exchange in market $(\fp_r,\fp_r,\cdots,\fp_r)$ will deviate to another reserve price to gain revenue. Theorem \ref{thm-allzero} settles the case of $r=0$, and here we only need to analyze the case of $r>0$. Since the other exchanges have the same reserve prices, we merge them and treat them as a single exchange. It suffices to prove when two exchanges propose $(\fp_r,\fp_r)$, there exists $\epsilon>0$ such that the first exchange can deviate to $\fp_{r-\epsilon}$ to get more revenue.

By equation (\ref{eqn-integral-ell}), when both exchanges use $\fp_{r}$, the bidding function is 
\begin{equation*}
\bidfn(v)=\begin{cases}
    0, &0\leq v<r;\\
    \frac{rg(r)+\int_{r}^{v}tg(t)dt}{G(v)}, &v\geq r.
\end{cases}
\end{equation*}
When exchange 1 switches to $\fp_{r-\epsilon}$, the bidding function $\bidf$ of the bidder will be as follows:
\begin{equation*}
\bidf(v)=\begin{cases}
    0, &0\leq v<r-\epsilon;\\
    \frac{(r-\epsilon)G(r-\epsilon)+\int_{r-\epsilon}^{v}tg(t)dt}{G(v)}, &r-\epsilon\leq v<a;\\
    \frac{rG(a)+\int_{a}^{v}tg(t)dt}{G(v)},&v\geq a.
\end{cases}
\end{equation*}
Here $a$ is the discontinuous point where bidder is indifferent in bidding $\lbidf(a)$ and $\rbidf(a)$, which is determined by $a-r=\flow_1(a-\lbidf(a))$. An upper bound of $a$ is $r+\frac{\flow_1}{1-\flow_1}\epsilon$, since $\lbidf(a)\geq r-\epsilon$. The same as before we define $v^{(1)}$ to be the winner's value. When exchange 1 switches from $\fp_r$ to $\fp_{r-\epsilon}$, he will gain revenue when $r-\epsilon\leq v^{(1)}<r$, and lose revenue when $v^{(1)}\geq r$. Conditioned on $v=v^{(1)}\in[r-\epsilon,r)$, the revenue gain is 
\[\bidf(v)-\bidfn(v)=\bidf(v)\geq r-\epsilon.\]
Conditioned on $v=v^{(1)}\in[r,a)$, the revenue loss is 
\begin{eqnarray*}
\bidfn(v)-\bidf(v)
&=&\frac{rG(r)+\int_{r}^{v}tg(t)dt}{G(v)}-\frac{(r-\epsilon)G(r-\epsilon)+\int_{r-\epsilon}^{v}tg(t)dt}{G(v)}\\
&=&\frac{rG(r)-(r-\epsilon)G(r-\epsilon)-\int_{r-\epsilon}^{r}tg(t)dt}{G(v)}\\
&\leq&\frac{rG(r)-(r-\epsilon)G(r-\epsilon)-\int_{r-\epsilon}^{r}(r-\epsilon)g(t)dt}{G(v)}\\
&=&\frac{G(r)}{G(v)}\epsilon\leq\frac{1}{G(v)}\epsilon.
\end{eqnarray*}
Conditioned on $v=v^{(1)}\in[a,\infty)$, the revenue loss is 
\begin{eqnarray*}
\bidfn(v)-\bidf(v)
&=&\frac{rG(r)+\int_{r}^{v}tg(t)dt}{G(v)}-\frac{rG(a)+\int_{a}^{v}tg(t)dt}{G(v)}\\
&=&\frac{\int_{r}^{a}(t-r)g(t)dt}{G(v)}\\
&\leq&\frac{g(a)}{G(v)}(a-r)^2\\
&=&\frac{(n-1)f(a)F^{n-2}(a)}{G(v)}(a-r)^2<\frac{nf(a)}{G(v)}(a-r)^2.
\end{eqnarray*}
Since $f(v)>0$ for $v>0$, we can find $c,\delta>0$ such that $f(v)\in[c,c+\delta]$ for $v\in[r-\epsilon,r+\frac{\flow_1}{1-\flow_1}\epsilon]$, and $\delta\to0$ when $\epsilon \to 0$. Since winner's value $v^{(1)}$ has cumulative density function $F^n(v)$ and density $nF^{n-1}(v)f(v)=nf(v)G(v)$, the total revenue gain is at least
\begin{eqnarray*}
\int_{r-\epsilon}^{\infty}nf(t)G(t)(\bidf(t)-\bidfn(t))
&>&\int_{r-\epsilon}^{r}nf(t)G(t)(r-\epsilon)dt-\int_{r}^{a}nf(t)G(t)\cdot\frac{1}{G(t)}\epsilon dt\\
& &-\int_{a}^{\infty}nf(t)G(t)\cdot\frac{nf(a)}{G(t)}(a-r)^2dt\\
&\geq&\int_{r-\epsilon}^{r}ncG(r-\epsilon)(r-\epsilon)dt-\int_{r}^{a}n(c+\delta)\epsilon dt\\
& &-\int_{0}^{\infty}f(t)dt\cdot n^2(a-r)^2(c+\delta)\\
&=&ncG(r-\epsilon)\cdot(r-\epsilon)\epsilon-n(c+\delta)(a-r)\epsilon\\
& &-n^2(a-r)^2(c+\delta)\\
&\geq&ncF^{n-1}(r-\epsilon)\cdot(r-\epsilon)\epsilon-n(c+\delta)\frac{\flow_1}{1-\flow_1}\epsilon^2\\
& &-n^2(c+\delta)\left(\frac{\flow_1}{1-\flow_1}\right)^2\epsilon^2\\
&\geq& 0
\end{eqnarray*}
when $\epsilon\to0$. The first inequality is true by applying previous bounds of revenue gain. The second inequality is true by $f(t)\in[c,c+\delta)$ for $t\in[r-\epsilon,a]$, and $g(s)=(n-1)f(s)F^{n-2}(s)\leq (n-1)(c+\delta)F^{n-2}(a)$ for $s\in[r,a]$. The third equality is true by $\int_{0}^{\infty}f(t)dt=1$. The fourth inequality is true by $a\leq r+\frac{\flow_1}{1-\flow_1}\epsilon$. The last inequality is true since when $\epsilon\to 0$, the lowest-degree term has positive coefficient. Thus there exists $\epsilon>0$ such that the first exchange prefers to switch to $\fp_{r-\epsilon}$.

\end{proof}



\bibliographystyle{plainnat}
\bibliography{reference}

\appendix
\section{Appendix}
\begin{proof}[Proof of Lemma \ref{lem-existence}]
We first prove that there exists a unique equilibrium that satisfy equation (\ref{eqn-integral-ell}) for each $\ell$. Suppose that we have already found a bidding function $\beta$ that satisfy (\ref{eqn-integral-ell}) for $<\ell$. Now we want to solve $\bidf$ for $v\in[a_\ell,a_{\ell+1})$. Let $I\subseteq[\ell]$ be the set of exchanges with reserve at most $r_\ell$ that uses first-price auctions, $J=[\ell]\setminus I$ be the set of exchanges with reserve at most $r_\ell$ that uses second-price auctions. Then the payment function $\avgp_{\ell}$ for $v\in[a_\ell,a_{\ell+1})$ would be the sum of first-price payments and second-price payments, which is
\begin{eqnarray*}
\avgp_{\ell}(\bidf(v))
&=&\frac{1}{\sum_{j\leq \ell}\flow_j}\left(\sum_{j\in I}\flow_j\bidf(v)+\sum_{j\in J}\flow_j\E_{\values}\left[\max(\bidf(v^{(2)}),r_j)\right]\right)\\
&=&\frac{1}{\sum_{j\leq \ell}\flow_j}\left(\sum_{j\in I}\flow_j\bidf(v)+\sum_{j\in J}\flow_j\frac{1}{G(v)}\left(G(a_j)r_j+\int_{a_j}^{v}\bidf(t)g(t)dt\right)\right)\\
&=&\frac{1}{\sum_{j\leq \ell}\flow_j}\Bigg(\bidf(v)\bigg(\sum_{j\in I}\flow_j\bigg)+\frac{1}{G(v)}\sum_{j\in J}\flow_jG(a_j)r_j\\
& &\ \ \ \ \ \ \ \ \ \ \ \ \ \ \ \ \ \ \ +\frac{1}{G(v)}\bigg(\sum_{j\in J}\flow_j\int_{a_j}^{v}\bidf(t)g(t)dt\bigg)\Bigg).
\end{eqnarray*}
Here the second equality follows from equation \ref{eqn-second-bid}. Apply to \ref{eqn-integral-ell} and multiply both sides by $\sum_{j\leq \ell}\flow_j$, and denote $\flow_I=\sum_{j\in I}\flow_j$, $\flow_J=\sum_{j\in J}\flow_j$, we get
\begin{eqnarray*}
(\flow_I+\flow_J)\int_{a_\ell}^{v}tg(t)dt
&=&\left(\flow_I\bidf(v)G(v)+\sum_{j\in J}\flow_jG(a_j)r_j+\sum_{j\in J}\flow_j\int_{a_j}^{v}\bidf(t)g(t)dt\right)\\
& &-\left(\flow_I\bidf(a_\ell)G(a_\ell)+\sum_{j\in J}\flow_jG(a_j)r_j+\sum_{j\in J}\flow_j\int_{a_j}^{a_{\ell}}\bidf(t)g(t)dt\right)\\
&=&\flow_I\bidf(v)G(v)-\flow_I\bidf(a_\ell)G(a_\ell)+\lambda_J\int_{a_\ell}^{v}\bidf(t)g(t)dt.
\end{eqnarray*}
Take the derivative on both sides, we get
\[(\lambda_I+\lambda_J)vg(v)=\lambda_IG(v)\bidf'(v)+(\lambda_I+\lambda_J)\bidf(v)g(v).\]
If $\lambda_I=0$, the above equation admits a solution $\bidf(v)=v$. This solution is consistent with $\bidf(a_\ell)=r_{\ell}$, as if all exchanges with reserve at most $r_\ell$ use second-price auction, the bidding equilibrium would be identity function. If $\lambda_I>0$, such differential equation always have a solution that satisfies $\bidf(a_\ell)=r_{\ell}$. Thus in every case there will be a bidding function $\bidf$ that satisfy equation (\ref{eqn-integral-ell}) for all $\ell$.

Notice that above analysis only shows that bidding according to $\bidf$ is local optimal. Now we prove that bidding according to $\bidf$ is also global optimal, i.e. a bidder with value $v$ will prefer to bid $\bidf(v)$ rather than $\bidf(v')$, for some $v'\neq v$. Define $h_\ell(v)=\avgp_\ell(\bidf(v))$ be the expected payment of winner when $v$ is winner's value, $\Lambda_\ell=\sum_{j\leq \ell}\flow_j$ be the market occupation of first $\ell$ exchanges. For $s\in[a_{\ell},a_{\ell+1})$, let $u(s,v)=\Lambda(\ell)G(s)(v-\avgp_\ell(s))$ be the expected utility obtained by a bidder with value $v$ when bidding $\bidf(s)$. It remains to show that for any $v$, $u(s,v)$ is maximized when $s=v$. For any $v''<v<v'$, assume that $v''\in[a_{j},a_{j+1})$, $v'\in[a_{k},a_{k+1})$, $v\in[a_{\ell}$, $j\leq \ell\leq k$. We now prove that $u(v',v)\leq u(v,v)$, and $u(v'',v)\leq u(v,v)$. 

Now we prove $u(v',v)\leq u(v,v)$. Notice that for any $j$, bidding $\rbidf(a_j)$ and $\lbidf(a_j)$ leads to different allocation. For $\ell+1\leq j\leq k$, we will use $u(a_j,v)$ and $u(a_j^-,v)$ to denote such difference in utility. Actually we can observe that, 
\begin{eqnarray*}
u(a_j,v)-u(a_j^-,v)&=&u(a_j,a_j)-u(a_j^-,a_j)-\flow_j(a_j-v)\\
&=&-\flow_j(a_j-v)\leq 0.
\end{eqnarray*}
The second equality comes from the fact that a bidder with value $a_j$ is indifferent in bidding $\lbidf(a_j)$ and $\rbidf(a_j)$.
Then
\begin{eqnarray}
u(v,v)-u(v',v)
&=&u(v,v)-u(a_{\ell+1}^-,v)+u(a_{\ell+1}^-,v)-u(a_{\ell+2}^-,v)\nonumber\\
& &+\cdots+u(a_{k-1}^-,v)-u(a_{k}^-,v)+u(a_k^-,v)-u(v',v)\nonumber\\
&\geq&\left(u(v,v)-u(a_{\ell+1}^-,v)\right)+\left(u(a_{\ell+1},v)-u(a_{\ell+2}^-,v)\right)\nonumber\\
& &+\cdots+\left(u(a_{k-1},v)-u(a_{k}^-,v)\right)+\left(u(a_k,v)-u(v',v)\right)\label{eqn-difference}.
\end{eqnarray}
Now it suffices to show that each difference in the above sum is non-negative.
Define $w(t)=u(t,t)$. Observe that for $t\in[a_k,v']$,
\begin{eqnarray*}
0&=&\frac{\partial}{\partial s}\cdot\frac{u(s,t)}{\Lambda_{k}}\Big|_{s=t}\\
&=&g(t)(t-h(t))-G(t)h'(t)\\
&\geq&g(t)(t-h(t))-G(t)h'(t)+G(t)-G(v')\\
&=&w'(t)-G(v').
\end{eqnarray*}
Here the first line follows by $s=t$ is the local maximal of $u(s,t)$; the third inequality follows from $G(v')\geq G(t)$.
Integrate $w'(t)+G(v')$ from $a_k$ to $v'$, we get
\begin{eqnarray*}
0&\geq&\int_{a_k}^{v'}(w'(t)-G(v'))dt\\
&=&w(v')-w(a_k)-G(v')(v'-a_k)\\
&=&G(v')(v'-h(v'))-G(a_k)(a_k-h(a_k))-G(v')(v'-a_k)\\
&=&G(v')(a_k-h(v'))-G(a_k)(a_k-h(a_k))\\
&\geq&G(v')(v-h(v'))-G(a_k)(v-h(a_k))\\
&=&\frac{1}{\Lambda_k}\big(u(v',v)-u(a_k,v)\big).
\end{eqnarray*}
Here the fifth line follows from $a_k\geq v$. Thus $u(a_k,v)-u(v',v)\geq 0$, and the same way we can prove $u(v,v)-u(a_{\ell+1}^-,v)\geq 0$, $u(a_{\ell+1},v)-u(a_{\ell+2}^-,v)\geq 0$, $\cdots$, $u(a_{k-1},v)-u(a_{k}^-,v)\geq 0$. 

Apply above results to inequality (\ref{eqn-difference}) we prove $u(v,v)\geq u(v',v)$. The proof of $u(v'',v)\leq u(v,v)$ can be done in the same way and is completely symmetric.

\end{proof}

\end{document}